\documentclass[12pt]{article}
\usepackage[english]{babel}
\usepackage[utf8]{inputenc}
\usepackage[T1]{fontenc}
\usepackage{listings}
\usepackage{amsbsy}
\usepackage{amsfonts}
\usepackage{amsmath}
\usepackage{amssymb}
\usepackage{amsthm}
\usepackage{graphicx} 
\usepackage[pdftex]{color}
\usepackage{dsfont}
\usepackage{enumerate}
\usepackage{float}
\usepackage{calc, color, setspace}
\usepackage{fullpage}
\usepackage{indentfirst}
\usepackage{longtable}
\usepackage{multirow}
\usepackage{natbib}
\usepackage{xr-hyper}
\usepackage{hyperref,bookmark}
\usepackage{mathtools}

\usepackage{listings}
\usepackage{fancyvrb}
\usepackage{tcolorbox}

\newtheorem{theorem}{Theorem}[section]
\newtheorem{definition}{Definition}[section]

\newtheorem{remark}[theorem]{Remark}

\hypersetup{
  colorlinks=true,
  linkcolor=red,
  citecolor=blue,
  filecolor=red,
  urlcolor=teal,
}
\newcommand{\blind}{1}

\begin{document}

\def\spacingset#1{\renewcommand{\baselinestretch}%
{#1}\small\normalsize} \spacingset{1}


\if1\blind
{
  \title{\bf  Generalised Linear Models Driven by Latent Processes: Asymptotic Theory and Applications
  }
  \author{Wagner Barreto-Souza$^\star$\footnote{E-mail: \textcolor{teal}{\texttt{wagner.barreto-souza@ucd.ie}}}\,\,\, and\, Ngai Hang Chan$^\sharp$\footnote{E-mail: \textcolor{teal}{\texttt{nhchan@cityu.edu.hk}}}\hspace{.2cm}\\
    {\it \normalsize $^\star$School of Mathematics and Statistics, University College Dublin, Dublin 4, Republic of Ireland}\\
    {\it \normalsize $^\sharp$Department of Biostatistics, City University of Hong Kong, Hong Kong}}
  \maketitle
} \fi
\if0\blind
{
  \bigskip
  \bigskip
  \bigskip
  \begin{center}
    {\LARGE\bf Generalised Linear Models Driven by Latent Processes: Asymptotic Theory and Applications}
\end{center}
  \medskip
} \fi

\bigskip
\addtocontents{toc}{\protect\setcounter{tocdepth}{1}}

\begin{abstract}

This paper introduces a class of generalised linear models (GLMs) driven by latent processes for modelling count, real-valued, binary, and positive continuous time series. Extending earlier latent-process regression frameworks based on Poisson or one-parameter exponential family assumptions, we allow the conditional distribution of the response to belong to a bi-parameter exponential family, with the latent process entering the conditional mean multiplicatively. This formulation substantially broadens the scope of latent-process GLMs, for instance, it naturally accommodates gamma responses for positive continuous data, enables estimation of an unknown dispersion parameter via method of moments, and avoids restrictive conditions on link functions that arise under existing formulations. We establish the asymptotic normality of the GLM estimators obtained from the GLM likelihood that ignores the latent process, and we derive the correct information matrix for valid inference. In addition, we provide a principled approach to prediction and forecasting in GLMs driven by latent processes, a topic not previously addressed in the literature. We present two real data applications on measles infections in North Rhine-Westphalia (Germany) and paleoclimatic glacial varves, which highlight the practical advantages and enhanced flexibility of the proposed modelling framework.

\end{abstract}


\section{Introduction}\label{intro}

In a pioneer work,  \cite{zeg1988} introduced a regression model for counts driven by a weakly stationary latent process. Only the two first conditional moments (given the latent process) were specified, and a quasi-likelihood estimation procedure was considered for inferential purposes. In a similar vein but assuming a conditional Poisson distribution, \cite{davetal2000} explored a Poisson count time series $\{Y_t\}$ defined as 
\begin{eqnarray}\label{poisson}
	Y_t|\nu_t\sim\mbox{Poisson}(\mu_t\nu_t),\quad \mu_t=\exp({\bf x}_t^\top\boldsymbol\beta),
\end{eqnarray}	  
where $\nu_t$ is a log-normal AR(1) process with mean 1 and some finite variance, ${\bf x}_t$ is a vector of covariates, and $\boldsymbol\beta$ is the associated vector of regression coefficients. Among the asymptotic properties explored by the authors, they established conditions to ensure asymptotic normality of the GLM estimators based on the Poisson likelihood by ignoring the latent process. The model (\ref{poisson}) was also previously studied by \cite{chaled1995}, where a Monte Carlo EM algorithm was developed for estimation of the parameters. A two-step estimation procedure with the regression parameters being estimated from the marginal distribution, and the parameters of the latent process estimated via a composite likelihood approach, was proposed by \cite{sor2019}.
  
\cite{davwu2009} extended model  (\ref{poisson}) by replacing the Poisson assumption by the one-parameter exponential family (EF), with focus on the negative binomial case (with known dispersion parameter). The extended model assumes that $Y_t$ given a stationary strongly mixing latent process $\nu_t$ follows an one-parameter exponential family with conditional mean 
\begin{eqnarray}\label{ef_dav}
E(Y_t|\nu_t)=h({\bf x}_t^\top\boldsymbol\beta+\nu_t),
\end{eqnarray}  
where $h(\cdot)$ is the inverse of a link function such that the resulting GLM likelihood is concave and $E(h({\bf x}_t^\top\boldsymbol\beta+\nu_t))=h({\bf x}_t^\top\boldsymbol\beta)$. Under the above formulation, such a function exists, for instance, for Poisson, negative binomial, and Gaussian cases. The asymptotic normality of the GLM estimators was established based on the GLM likelihood by ignoring the latent process.
  
Other related recent contributions include \cite{maiaetal2021}, where a class of semiparametric time series was proposed by using quasi-likelihood models driven by latent processes, and \cite{baromb2022}, where a Poisson regression driven by a gamma AR(1) process was developed with a composite likelihood inferential approach.

Our chief goal in this paper is to introduce a flexible GLM driven by latent processes to handle counts, real-valued, continuous, positive continuous time series. To do this, we assume that the time series of interest, say $\{Y_t\}$, conditional on a latent process $\{\nu_t\}$, follows a bi-parameter exponential family, with a conditional mean having the latent process in a multiplicative way, and with a dispersion parameter, which can be unknown. We establish the asymptotic normality of the GLM estimators and provide the the correct information matrix to assess standard errors of the parameter estimates. Our formulation have some advantages when compared to the \cite{davwu2009}'s model as follows: (i) by entering the latent effect in a multiplicative way instead of additive like in (\ref{ef_dav}), it allows us to handle a gamma time series to address positive continuous time series, which is not feasible under the existing approach (more details on that are provided in Remark \ref{rem:comparison}); (ii) we consider a dispersion parameter that can be estimated based on the method of moments, so adding more flexibility; (iii) if the function $h(\cdot)$ is not exponential, then a multiplicative latent effect form is not possible and therefore it is hard to ensure the requirements by \cite{davwu2009} to establish the asymptotics for the GLM estimators such as  $E(h({\bf x}_t^\top\boldsymbol\beta+\nu_t))=h({\bf x}_t^\top\boldsymbol\beta)$ and their condition stated in Theorem 3 (weak convergence of their $C_n(s)$ quantity); under our formulation, these issues are easily addressed; (iv) we also address how to predict based on GLMs driven by latent processes, which is not addressed in the current literature;
(v) we consider different latent processes in our numerical results, while a log-normal AR(1) assumption has been usually adopted in the literature; for example, our empirical illustrations show that a gamma AR(1) latent process might provide better results when compared to the log-normal AR(1) process.

This paper is organised as follows. Section \ref{sec:modeldef} introduces our class of GLMs driven by latent processes and provide some basic results. Section \ref{sec:asymptotics} establishes the asymptotic normality of the GLM estimators based on the Central Limit Theorem for strongly mixing processes by \cite{pelute1997}, and explicitly provide the correct information matrix to assess the standard errors for valid inference. Prediction and forecasting are addressed in Section \ref{sec:prediction}. Section \ref{sec:applications} presents two real data applications on measles infection cases in North Rhine-Westphalia (Germany) and paleoclimatic glacial varves  time series, which illustrates the Poisson and gamma GLM time series models' performance in practice. Concluding remarks are presented in Section \ref{sec:conclusion}.


\section{Model specification}\label{sec:modeldef}

We start this section by introducing key ingredients to define our class of generalised linear models driven by latent processes.
Consider a random variable $Y$ with distribution belonging to the exponential family (EF) and density/probability function given by
\begin{eqnarray*}
\pi(y)=\exp\left\{\dfrac{\theta y-b(\theta)}{\phi}+c(y;\phi)\right\},\quad y\in\mathbb S,	
\end{eqnarray*}	
where $b(\cdot)$ is assumed to be twice differentiable, $\theta$ is a real-valued parameter, $\phi>0$ is a dispersion parameter, and $\mathbb S$ is the support of the distribution. It is well-know from the EF theory that $\mu\equiv E(Y)=b'(\theta)$ and $\mbox{Var}(Y)=\phi b''(\theta)=\phi V(\mu)$, where $b'(\cdot)$ and $b''(\cdot)$ denote respectively the first and second order derivatives of $b(\cdot)$, and $V(\cdot)$ denotes the variance function. The parameter $\theta$ can be expressed in terms of $\mu$,  say $\theta=g(\mu)$. We will also adopt the following notation $Y\sim\mbox{EF}(\mu,\phi)$. We now describe the latent stationary and strongly mixing processes that will be considered in our numerical experiments.\\

\noindent {\bf Log-normal AR(1)}. Define $\{\nu_t\}_{t\in\mathbb N}$ by $\nu_t=\exp(Z_t)$, where $\{Z_t\}_{t\in\mathbb N}$ be an AR(1) process with $N(-\sigma^2/2,\sigma^2)$ marginals and AR parameter $\rho\in(-1,1)$. Then, $\nu_t$ has log-normal marginals, $E(\nu_t)=1$, $\mbox{Var}(\nu_t)=\exp(\sigma^2)-1$, and $\mbox{corr}(\nu_{t+l},\nu_t)=\dfrac{\exp(\sigma^2\rho^l)-1}{\exp(\sigma^2)-1}$, for $l\geq1$.\\

\noindent {\bf Gamma AR(1)}.  The sequence $\{\nu_t\}_{t\in\mathbb N}$ follows a gamma AR(1) (in short GAR) process by \cite{sim1990} with mean 1 and variance $\sigma^2>0$ if the conditional density function of $\nu_t$ given  $\nu_{t-1}$ is given by
\begin{eqnarray}\label{gar(1)}
f(\nu_t|\nu_{t-1})=\dfrac{1}{\sigma^2(1-\rho)}\left(\dfrac{\nu_t}{\rho\nu_{t-1}}\right)^{\frac{\sigma^{-2}-1}{2}}\exp\left(-\dfrac{\nu_t+\rho\nu_{t-1}}{\sigma^2(1-\rho)}\right)\mathcal I_{\sigma^{-2}-1}\left(2\dfrac{\sqrt{\rho\nu_t\nu_{t-1}}}{\sigma^2(1-\rho)}\right),
\end{eqnarray}
where $\rho\in(0,1)$ controls the time dependence, and $\mathcal I_{u}(x)=\displaystyle\sum_{k=0}^\infty\dfrac{(x/2)^{2k+u}}{\Gamma(k+u+1)k!}$ is the modified Bessel function of the first kind of order $u\in\mathbb R$. Also, the GAR process is strictly stationary having gamma marginal distributions with shape and scale parameters equal to $\sigma^{-2}$, and autocorrelation function given by $\mbox{corr}(\nu_{t+l},\nu_t)=\rho^l$, for $l\geq1$.\\ 

\noindent {\bf Squared ARCH(1)}.  Let $\{Z_t\}_{t\in\mathbb N}$ be an ARCH(1) process \citep{eng1982}, that is 
\begin{eqnarray}\label{arch}
Z_t&=&\sqrt{\omega+\rho Z_{t-1}^2}\epsilon_t, 
\end{eqnarray}
where $\omega>0$, $\rho\in(0,1)$, and $\{\epsilon_t\}_{t\in\mathbb N}\stackrel{iid}{\sim}N(0,1)$. Define the squared ARCH(1) process by $\nu_t=Z_t^2$, for $t\in\mathbb N$, with $E(\nu_t)=1$. This will give us that $\omega=1-\rho$ and will be in force for what follows. This condition will later ensure that the model parameters can be consistently estimated. Moreover, we assume $\rho\in(0,1/\sqrt{3})$ so that the second moment of $\nu_t$ is finite and given by $E(\nu_t^2)=3\dfrac{1-\rho^2}{1-3\rho^2}$. In this case, the autocorrelation function of $\nu_t$ is $\mbox{corr}(\nu_{t+l},\nu_t)=\rho^l$, for $l\geq1$.\\

All ingredients have been established to define the model. Let  $\{\nu_t\}_{t\in\mathbb N}$ be a latent process with $E(\nu_t)=1$ and $E(\nu^2_t)<\infty$. We consider a time series $\{Y_t\}_{t\in\mathbb N}$ with conditional distribution given $\{\nu_t\}_{t\in\mathbb N}$ belonging to the exponential family:
\begin{eqnarray}\label{glm1}
	Y_t|\nu_t\sim \mbox{EF}(\mu_t\nu_t,\phi),
\end{eqnarray}	
with  
\begin{eqnarray}\label{glm2}
	\mu_t\equiv h({\bf x}_{nt}^\top\boldsymbol{\beta}),
\end{eqnarray}	
where $h(\cdot)$ is the inverse of a canonical link function, ${\bf x}_{n\,t}$ is a $p\times 1$ vector of covariates that might depend on the sample size $n$, $\boldsymbol\beta$ is its associated coefficient vector with dimension $p\times 1$, and $\phi>0$ is a dispersion/precision parameter.

\begin{remark}\label{rem:comparison}
Our formulation (\ref{glm1}) and (\ref{glm2}) is different from that considered by \cite{davwu2009} (see Equation \ref{ef_dav}). We enter the latent process into the model in a multiplicative way. Our formulation and \cite{davwu2009} ones will coincide only when the log link function is considered. Under (\ref{glm1}) and (\ref{glm2}), we are able to handle continuous positive time series based on, for example, the gamma and inverse Gaussian distributions since the conditions on the link function described below Equation (\ref{ef_dav}) are satisfied. On the other hand, such conditions for the gamma and inverse Gaussian cases are complicated (if not impossible) to be checked under the formulation by \cite{davwu2009}.
\end{remark}

The expected value of our time series is $$E(Y_t)=E[E(Y_t|\nu_t)]=h({\bf x}_{nt}^\top\boldsymbol{\beta})E(\nu_t)=h({\bf x}_{nt}^\top\boldsymbol{\beta})=\mu_t$$ since we are imposing that $E(\nu_t)=1$. By assuming that the variance function of the GLM assumes the form $V(\mu)=\mu^\gamma$, for $\gamma\geq0$, and that $E(\nu_t^{\max(\gamma,2)})<\infty$, we obtain that the variance of $Y_t$ is
\begin{eqnarray*}
\mbox{Var}(Y_t)&=&E[\mbox{Var}(Y_t|\nu_t)]+\mbox{Var}[E(Y_t|\nu_t)]=E[\phi V(\mu_t \nu_t)]+\mbox{Var}(\mu_t \nu_t)\\
&=&\phi \mu_t^\gamma E(\nu_t^\gamma)+\mu_t^2\mbox{Var}(\nu_t)=\phi \mu_t^\gamma \kappa_{\gamma}+\mu_t^2(\kappa_2-1),
\end{eqnarray*}
where $\kappa_j\equiv E(\nu_t^j)$ for $j>0$. For $l\in\mathbb N$, the covariance function is 
\begin{eqnarray*}
\mbox{cov}(Y_{t+l},Y_t)&=&\mbox{cov}(E(Y_{t+l}|\nu_{t+l}),E(Y_t|\nu_t))+E(\mbox{cov}(Y_{t+l},Y_t|\nu_{t+l},\nu_t))\\
&=&\mu_{t+l}\mu_t\mbox{cov}(\nu_{t+l},\nu_t)+0=\mu_{t+l}\mu_t\mbox{cov}(\nu_{t+l},\nu_t),
\end{eqnarray*}
and the autocorrelation function assumes the form
\begin{eqnarray*}
	\mbox{corr}(Y_{t+l},Y_t)&=&\frac{\mu_{t+l}\mu_t\mbox{cov}(\nu_{t+l},\nu_t)}{\sqrt{[\phi \mu_{t+l}^\gamma \kappa_{\gamma}+\mu_{t+l}^2(\kappa_2-1)][\phi \mu_t^\gamma \kappa_{\gamma}+\mu_t^2(\kappa_2-1)]}}.
\end{eqnarray*}

Explicit expressions for the variance and autocovariance/autocorrelation function of $Y_t$ are obtained by using the moments and autocovariance/autocorrelation function of the latent processes described above.


\section{Asymptotic distribution of the GLM estimators}\label{sec:asymptotics}

We now establish the asymptotic distribution of the estimators based on a GLM pseudo log-likelihood function for the GLMs driven by latent processes. The estimation of $\phi$ and the parameters related to the latent process will be performed via method of moments and will be illustrated in our real data applications in Section \ref{sec:applications}.

The pseudo log-likelihood function under a GLM with  canonical link function (hence $\theta_t=g({\bf x}_{n\,t}^\top\boldsymbol\beta)={\bf x}_{n\,t}^\top\boldsymbol\beta$) is given by
\begin{eqnarray*}
\ell(\boldsymbol\beta)=\sum_{t=1}^n\left\{[{\bf x}_{nt}^\top\boldsymbol\beta Y_t-b({\bf x}_{nt}^\top\boldsymbol\beta)]/\phi+c(Y_t;\phi)\right\}.
\end{eqnarray*}

The GLM estimator of $\boldsymbol\beta$, say $\widehat{\boldsymbol\beta}$, is obtained by $\widehat{\boldsymbol\beta}=\mbox{argmax}_{\boldsymbol\beta}\ell(\boldsymbol\beta)$, which does not depend on $\phi$. Since the canonical link function is considered, $\ell(\cdot)$ is concave, and this ensures the existence and uniqueness of $\widehat{\boldsymbol\beta}$. To show the asymptotic normality of the GLM estimators, we start by showing the asymptotic normality of the quantity
\begin{eqnarray}\label{eq:key}
n^{-1/2}\sum_{t=1}^n{\bf x}_{nt}\mu_t^0(\nu_t-1), 
\end{eqnarray}	
which is a key result to establish the asymptotic distribution of $\sqrt{n}(\widehat{\boldsymbol\beta}-\boldsymbol\beta_0)$, where $\boldsymbol\beta_0$ is the true value of $\boldsymbol\beta$ and $\mu_t^0=h({\bf x}_{nt}^\top\boldsymbol\beta_0)$.

To do this, we need the definition of $\alpha$-mixing (strongly mixing) processes and a Central Limit Theorem (CLT) by \cite{pelute1997}.

\begin{definition}\label{alpha-mix}[$\alpha$-mixing]  Let $\{X_t\}_{t\in\mathbb N}$ be a stationary process and denote the $\sigma$-algebra $\mathcal F_n^m=\sigma(X_n,X_{n+1},\ldots,X_m)$, for $n\leq m$. Define $\alpha(n)=\sup_k\sup_{A\in\mathcal F_1^k, B\in\mathcal F_{k+n}^\infty}|P(A\cap B)-P(A)P(B)|$. We say that $\{X_t\}_{t\in\mathbb N}$  is $\alpha$-mixing (strongly mixing) if $\alpha(n)\rightarrow0$ as $n\rightarrow\infty$.
\end{definition}

Let $\{a_{nt}\}_{n,t\in\mathbb N}$ be a triangular array of real numbers. Theorem 2.2(c) from \cite{pelute1997} states that a centred stochastic sequence $\{\zeta_t\}_{t\in\mathbb N}$ satisfies the weak convergence $\displaystyle\sum_{t=1}^na_{nt}\zeta_t\stackrel{d}{\longrightarrow}N(0,1)$ as $n\rightarrow\infty$ if the following conditions hold:

(C1) $\{\zeta_t\}_{t\in\mathbb N}$ is strongly mixing, and for certain $\delta>0$, $\{|\zeta_t|^{2+\delta}\}_{t\in\mathbb N}$ is uniformly integrable, $\inf_t\mbox{Var}(\zeta_t)>0$ and $\displaystyle\sum_{t=1}^\infty t^{2/\delta}\alpha(t)<\infty$;

(C2) $\{\zeta^2_t\}_{t\in\mathbb N}$ is an uniformly integrable family and $\mbox{Var}\left(\displaystyle\sum_{t=1}^na_{nt}\zeta_t\right)=1$;
and

(C3) $\sup_n\displaystyle\sum_{t=1}^n a_{nt}^2<\infty$ and $\max_{1\leq t\leq n}|a_{nt}|\rightarrow0$ as $n\rightarrow\infty$.\\

	
We also need to impose some conditions on the covariates to derive the asymptotic properties for the GLM estimators. The conditions below have been considered by \cite{davetal2000} and \cite{davwu2009}.\\

\noindent {\bf Assumptions on covariates.} Denote by $\boldsymbol\beta_0$ the true value of $\boldsymbol\beta$ and $\mu_t\equiv h({\bf x}_{nt}^\top\boldsymbol\beta_0)$.

\noindent  There exists a non-singular matrix $\boldsymbol\Omega_I$ such that
\begin{eqnarray}\label{as:b2}
	\lim_{n\rightarrow\infty}n^{-1}\sum_{t=1}^n{\bf x}_{nt}{\bf x}_{nt}^\top V(\mu_t^0)\equiv\boldsymbol\Omega_I.
\end{eqnarray}	
There exists a matrix $\boldsymbol\Omega^\dag_I$ such that the following convergence in probability holds as $n\rightarrow\infty$:
\begin{eqnarray}\label{as:varfunc}
	n^{-1}\sum_{t=1}^n{\bf x}_{nt} {\bf x}_{nt}^\top V(\mu_t^0\nu_t)\stackrel{p}{\longrightarrow}\boldsymbol\Omega^\dag_I.
\end{eqnarray}	
As $n\rightarrow\infty$, we assume that 
\begin{eqnarray}\label{as:mutmulcentre}
	n^{-1}\sum_{t=1}^n{\bf x}_{nt} {\bf x}_{nt}^\top\mu_t^0\mu_{t+l}^0\longrightarrow {\bf W}_l,
\end{eqnarray}	
uniformly in $|l|<n$,
\begin{eqnarray}\label{as:mutmulleft}
	n^{-1}\sum_{t=1}^{-l}{\bf x}_{nt} {\bf x}_{nt}^\top\mu_t^0\mu_{t+l}^0\longrightarrow 0,
\end{eqnarray}	
for each $l<0$, and
\begin{eqnarray}\label{as:mutmulright}
	n^{-1}\sum_{t=n-l+1}^n{\bf x}_{nt} {\bf x}_{nt}^\top\mu_t^0\mu_{t+l}^0\longrightarrow 0,
\end{eqnarray}	
for each $l>0$. Furthermore, it is assumed that left-hand sides of (\ref{as:mutmulleft}) and (\ref{as:mutmulright}) are uniformly bounded in $l\in(-n,0)$ and $l\in(0,n)$, respectively, as $n\rightarrow\infty$. The last assumption on the covariates is that
\begin{eqnarray}\label{as:sup}
	\sup_{1\leq t\leq n}n^{-1/2}|{\bf x}_{nt}\mu_t^0|\rightarrow0,
\end{eqnarray}	
as $n\rightarrow\infty$.

In the next theorem, we provide the asymptotic distribution of (\ref{eq:key}).

\begin{theorem}\label{thm:aux}
	Assume that $\{\nu_t\}_{t\in\mathbb N}$ is a stochastic process satisfying (C1) and that the assumptions on covariates (\ref{as:mutmulcentre})-(\ref{as:sup}) hold. Then, the following asymptotic normality is valid:
	\begin{eqnarray*}
		n^{-1/2}\sum_{t=1}^n{\bf x}_{nt}\mu_t^0(\nu_t-1)\stackrel{d}{\longrightarrow}N\left({\bf 0},{\boldsymbol\Omega}_{II}\right),
	\end{eqnarray*}	
	as $n\rightarrow\infty$, where $\boldsymbol\beta_0$ denotes the true value of $\boldsymbol\beta$ and $\boldsymbol\Omega_{II}=\displaystyle\sum_{l=-\infty}^\infty\gamma_\nu(l){\bf W}_l$,
	with $\gamma_\nu(\cdot)$ being the autocovariance function of $\nu_t$.
\end{theorem}	

\begin{proof}
Define $C_n({\bf s})={\bf s}^\top n^{-1/2}\displaystyle\sum_{t=1}^n{\bf x}_{nt}\mu_t^0(\nu_t-1)$ and $\tau^2_n({\bf s})=\mbox{Var}(C_n({\bf s}))$. It follows that
\begin{eqnarray*}
\tau^2_n({\bf s})&=&{\bf s}^\top n^{-1}\sum_{t=1}^n\sum_{k=1}^n{\bf x}_{nt}\mu_t^0{\bf x}_{nk}^\top\mu_k^0\mbox{cov}(\nu_t,\nu_k){\bf s}\\
&=&{\bf s}^\top n^{-1}\sum_{t=1}^n\sum_{k=1}^n{\bf x}_{nt}{\bf x}_{nk}^\top\mu_t^0\mu_k^0\gamma_\nu(t-k){\bf s},
\end{eqnarray*}	
By making a change of index $l=t-k$ and using Assumptions (\ref{as:mutmulcentre}), (\ref{as:mutmulleft}), and (\ref{as:mutmulright}), we obtain that
\begin{eqnarray*}
\lim_{n\rightarrow\infty}\tau^2_n({\bf s})={\bf s}^\top\sum_{l=-\infty}^\infty\gamma_\nu(l){\bf W}_l{\bf s}={\bf s}^\top{\boldsymbol\Omega}_{II}{\bf s}.
\end{eqnarray*}	

We now define $\zeta_t=\nu_t-1$ and $a_{nt}=\dfrac{1}{\tau_n({\bf s})}n^{-1/2}{\bf s}^\top{\bf x}_{nt}\mu_t^0$. Condition (C1) from the CLT for strongly mixing processes \citep{pelute1997} is a hypothesis of our theorem. Also, $\mbox{Var}\left(\displaystyle\sum_{t=1}^na_{nt}\zeta_t\right)=\dfrac{1}{\tau^2_n({\bf s})}\mbox{Var}(C_n({\bf s}))=\dfrac{1}{\tau^2_n({\bf s})}\tau^2_n({\bf s})=1$. Therefore, Condition (C2) is satisfied. Moreover, 
\begin{eqnarray*}
\sum_{t=1}^na_{nt}^2=\dfrac{1}{\tau^2_n({\bf s})}{\bf s}^\top n^{-1}\sum_{t=1}^n{\bf x}_{nt}{\bf x}_{nt}^\top(\mu_t^0)^2{\bf s}\longrightarrow\dfrac{{\bf s}^\top{\bf W}_0{\bf s}}{{\bf s}^\top{\bf W}{\bf s}},
\end{eqnarray*}	
as $n\rightarrow\infty$, and
\begin{eqnarray*}
	\max_{1\leq t\leq n}|a_{nt}|=\dfrac{1}{\tau_n({\bf s})}\max_{1\leq t\leq n}|{\bf s}^\top n^{-1/2}{\bf x}_{nt}\mu_t^0|\leq
	\dfrac{1}{\tau_n({\bf s})}|{\bf s}|^\top \max_{1\leq t\leq n}n^{-1/2}| {\bf x}_{nt}\mu_t^0|\longrightarrow0,
\end{eqnarray*}	
as $n\rightarrow\infty$, where we used that $\lim_{n\rightarrow\infty}\tau_n({\bf s})=\sqrt{{\bf s}^\top{\bf W}{\bf s}}$ and Assumption (\ref{as:sup}) to obtain the null limit. These results give us that Condition (C3) holds. Since all conditions from Theorem 2.2(c) from \cite{pelute1997} hold, we obtain the claimed asymptotic normality.
\end{proof}

\begin{theorem}
	Suppose that $\{\nu_t\}_{t\in\mathbb N}$ is a stochastic process satisfying (C1) and that the covariates satisfy Assumptions (\ref{as:b2})-(\ref{as:sup}). Then, the GLM estimator of $\boldsymbol\beta$ is consistent and satisfies the asymptotic normality
	\begin{eqnarray*}
		\sqrt{n}\left(\widehat{\boldsymbol\beta}-\boldsymbol\beta_0\right)\stackrel{d}{\longrightarrow}N\left({\bf 0},\boldsymbol\Omega_I^{-1}(\phi{\boldsymbol\Omega}_I^{\dag}+{\boldsymbol\Omega}_{II}){\boldsymbol\Omega}_I^{-1}\right),
	\end{eqnarray*}	
as $n\rightarrow\infty$, where $\boldsymbol\beta_0$ denotes the true value of $\boldsymbol\beta$.
\end{theorem}

\begin{proof}

Following a similar strategy by \cite{davetal2000} and \cite{davwu2009}, we define ${\bf u}\equiv \sqrt{n}(\boldsymbol\beta-\boldsymbol\beta_0)$ and $g_n({\bf u})\equiv-\phi\{\ell(\boldsymbol\beta)-\ell(\boldsymbol\beta_0)\}$, where $\boldsymbol\beta_0$ is the true parameter vector. Under this formulation, the GLM estimator is obtained by minimising $g_n({\bf u})$ with respect to $\boldsymbol\beta$.
Also, $\boldsymbol\beta=n^{-1/2}{\bf u}+\boldsymbol\beta_0$. It follows that
\begin{eqnarray*}
	g_n({\bf u})&=&-\phi\{\ell(n^{-1}{\bf u}+\boldsymbol\beta_0)-\ell(\boldsymbol\beta_0)\}\\
	&=&\sum_{t=1}^n\left\{
	{\bf x}^\top_{nt}\boldsymbol\beta_0Y_t-b({\bf x}^\top_{nt}\boldsymbol\beta_0)-{\bf x}^\top_{nt}(n^{-1/2}{\bf u}+\boldsymbol\beta_0)Y_t
	+b\left({\bf x}^\top_{nt}(n^{-1/2}{\bf u}+\boldsymbol\beta_0)\right)
	\right\}\\
	&=&\sum_{t=1}^n\left\{-{\bf x}^\top_{nt}n^{-1/2}{\bf u}Y_t+b\left({\bf x}^\top_{nt}(n^{-1/2}{\bf u}+\boldsymbol\beta_0)\right)-b({\bf x}^\top_{nt}\boldsymbol\beta_0)
	\right\}.
\end{eqnarray*}	

We have that 
\begin{eqnarray*}
\sum_{t=1}^n\left\{b\left({\bf x}^\top_{nt}(n^{-1/2}{\bf u}+\boldsymbol\beta_0)\right)-
b({\bf x}^\top_{nt}\boldsymbol\beta_0)\right\}\hspace{-.3cm}&\approx&\hspace{-.3cm} \sum_{t=1}^n \left\{{\bf x}^\top_{nt} n^{-1/2}{\bf u}b'({\bf x}^\top_{nt}\boldsymbol\beta_0)+\dfrac{({\bf x}^\top_{nt} n^{-1/2}{\bf u})^2}{2}b''({\bf x}^\top_{nt}\boldsymbol\beta_0)\right\}\\
&=&\sum_{t=1}^n \left\{{\bf x}^\top_{nt} n^{-1/2}{\bf u}\mu_t^0+\dfrac{({\bf x}^\top_{nt} n^{-1/2}{\bf u})^2}{2}V(\mu_t^0)\right\},
\end{eqnarray*}
where we have defined $\mu_t^0\equiv h({\bf x}^\top_{nt}\boldsymbol\beta_0)$, and $V(\mu_t^0)\equiv b''({\bf x}^\top_{nt}\boldsymbol\beta_0)$ denotes the variance function from GLMs.
It follows that 
\begin{eqnarray}\label{eq:g_approx}
	g_n({\bf u})\stackrel{p}{\approx} -n^{-1/2}\sum_{t=1}^n{\bf x}^\top_{nt}{\bf u}(Y_t-\mu_t^0)+\dfrac{1}{2}{\bf u}^\top\left(n^{-1}\sum_{t=1}^n{\bf x}_{nt}{\bf x}_{nt}^\top V(\mu_t^0)\right){\bf u}.
\end{eqnarray}	

From Assumption (\ref{as:b2}), the second term on the right side of (\ref{eq:g_approx}) converges to $\dfrac{1}{2}{\bf u}^\top\boldsymbol\Omega_I{\bf u}$ as $n\rightarrow\infty$. We will now establish the asymptotic normality of ${\bf Q}_n\equiv n^{-1/2}\displaystyle\sum_{t=1}^n(Y_t-\mu_t^0){\bf x}_{nt}$.

For a real vector ${\bf s}$, the conditional characteristic function of ${\bf Q}_n$ given $\nu_t$ is
\begin{eqnarray*}
E\left(\exp\{i{\bf s}^\top {\bf Q}_n\}|\nu_t\right)=\exp\left\{-i{\bf s}^\top n^{-1/2}\sum_{t=1}^n\mu_t^0{\bf x}_{nt}\right\}	
E\left(\exp\left\{-i{\bf s}^\top n^{-1/2}\sum_{t=1}^nY_t{\bf x}_{nt}\right\}\Big|\nu_t\right),
\end{eqnarray*}	
where we use the model specification $Y_t|\nu_t\sim\mbox{EF}(\mu_t\nu_t,\phi)$ and the notation $\widetilde\theta_t=g(\mu_t\nu_t)$ to obtain that
\begin{eqnarray*}
	E\left(\exp\left\{-i{\bf s}^\top n^{-1/2}\sum_{t=1}^nY_t{\bf x}_{nt}\right\}\Big|\nu_t\right)=\exp\left\{\phi^{-1}\sum_{t=1}^n\left[b(\widetilde\theta_t+i\phi n^{-1/2}{\bf s}^\top{\bf x}_{nt})-b(\widetilde\theta_t)\right]\right\}\\
	\stackrel{p}{\approx} \exp\left\{\phi^{-1}\sum_{t=1}^n\left[i{\bf s}^\top{\bf x}_{nt}n^{-1/2}\phi b'(\widetilde\theta_t)+\dfrac{(i{\bf s}^\top{\bf x}_{nt}n^{-1/2}\phi)^2}{2} b''(\widetilde\theta_t)\right]\right\}\\
	= \exp\left\{i{\bf s}^\top n^{-1/2}\sum_{t=1}^n{\bf x}_{nt}\mu_t^0\nu_t-\dfrac{1}{2}\phi{\bf s}^\top\left(n^{-1}\sum_{t=1}^n{\bf x}_{nt} {\bf x}_{nt}^\top V(\mu_t^0\nu_t)\right){\bf s}\right\}.
\end{eqnarray*}	

Hence, we obtain that
\begin{eqnarray*}
	E\left(\exp\{i{\bf s}^\top {\bf Q}_n\}|\nu_t\right)\stackrel{p}{\approx} \exp\left\{i{\bf s}^\top n^{-1/2}\sum_{t=1}^n{\bf x}_{nt}\mu_t^0(\nu_t-1)-\dfrac{1}{2}\phi{\bf s}^\top\left(n^{-1}\sum_{t=1}^n{\bf x}_{nt} {\bf x}_{nt}^\top V(\mu_t^0\nu_t)\right){\bf s}\right\}.
\end{eqnarray*}	

From Theorem \ref{thm:aux}, $n^{-1/2}\displaystyle\sum_{t=1}^n{\bf x}_{nt}\mu_t^0(\nu_t-1)\stackrel{d}{\rightarrow}{\bf F}\sim N\left({\bf 0},{\boldsymbol\Omega}_{II}\right)$, and from Assumption \ref{as:varfunc} $n^{-1}\displaystyle\sum_{t=1}^n{\bf x}_{nt} {\bf x}_{nt}^\top V(\mu_t^0\nu_t)\stackrel{p}{\longrightarrow}\boldsymbol\Omega^\dag_I$. Hence,
\begin{eqnarray*}
	\lim_{n\rightarrow\infty}E\left(\exp\{i{\bf s}^\top {\bf Q}_n\}\right)=E\left(\exp\left\{i{\bf s}^\top {\bf F}-\dfrac{1}{2}\phi{\bf s}^\top\boldsymbol\Omega^\dag_I{\bf s}\right\}\right)=\exp\left\{-\dfrac{1}{2}{\bf s}^\top\left(\phi\boldsymbol\Omega^\dag_I+\boldsymbol\Omega_{II}\right){\bf s}\right\},
\end{eqnarray*}	
that is ${\bf Q}_n\stackrel{d}{\rightarrow}{\bf Q}\sim N({\bf0},\phi\boldsymbol\Omega^\dag_I+\boldsymbol\Omega_{II})$. This implies that
\begin{eqnarray*}
	g_n({\bf u})\stackrel{d}{\rightarrow}-{\bf u}^\top{\bf Q}+\dfrac{1}{2}{\bf u}^\top\boldsymbol\Omega_I{\bf u}.
\end{eqnarray*}	

Arguing similarly as in \cite{davwu2009}, the limit process of $g_n({\bf u})$ has a unique minimiser, say $\widetilde{\bf u}$, and $	\sqrt{n}\left(\widehat{\boldsymbol\beta}-\boldsymbol\beta_0\right)\stackrel{d}{\longrightarrow}\widetilde{\bf u}\sim N\left({\bf 0},\boldsymbol\Omega_I^{-1}(\phi{\boldsymbol\Omega}_I^{\dag}+{\boldsymbol\Omega}_{II}){\boldsymbol\Omega}_I^{-1}\right)$.
\end{proof}


\section{Prediction and forecasting}\label{sec:prediction}


We here propose to perform prediction or forecasting based on the conditional expectation $E(Y_{t+h}|Y_t)$. Although our proposed GLM time series are not Markovian, this can be used for prediction purposes. Such approach has been successfully employed in non-Markovian models; for instance, see \cite{maiaetal2021} and  \cite{baromb2022}. Following similar arguments as in those papers, for $l\geq1$, we obtain that
\begin{eqnarray}\label{cond_exp}
E(Y_{t+l}|Y_t)=E[E(Y_{t+l}|\nu_t)|Y_t],
\end{eqnarray}	
with $E(Y_{t+l}|\nu_t)=E[E(Y_{t+l}|\nu_{t+l})|\nu_t]=\mu_{t+l}E(\nu_{t+l}|\nu_t)$. We will now provide explicit expressions for (\ref{cond_exp}) for the three latent processes considered in this paper. Under a log-normal AR(1) process, it follows that $E(\nu_{t+l}|\nu_t)=\exp(\rho^l\sigma^2(1-\rho^l)/2+\rho^l\log\nu_t)$. Hence,
\begin{eqnarray*}
	E(Y_{t+l}|Y_t)=\mu_{t+l}\exp(\rho^l\sigma^2(1-\rho^l)/2)E[\exp(\rho^l\log\nu_t)|Y_t)],
\end{eqnarray*}	
with 
\begin{eqnarray*}
E[\exp(\rho^l\log\nu_t)|Y_t)]=\dfrac{\displaystyle\int_{0}^\infty\exp(\rho^l\log\nu)f_{y_t|\nu_t}(y|\nu)f_{\nu_t}(\nu)d\nu}{\displaystyle\int_{0}^\infty f_{y_t|\nu_t}(y|\nu)f_{\nu_t}(\nu)d\nu},
\end{eqnarray*}	
which can be computed numerically, where $f_{y_t|\nu_t}(\cdot|\cdot)$ and  $f_{\nu_t}(\nu)$ are $\mbox{EF}$ and log-normal density functions, respectively. For $\nu_t\sim\mbox{GAR}(1)$, we have that $E(\nu_{t+l}|\nu_t)=1+\rho^l(\nu_1-1)$ (for instance, see \cite{baromb2022}), and therefore
\begin{eqnarray}\label{cond_exp_gar}
	E(Y_{t+l}|Y_t)=\mu_{t+l}\{1+\rho^l[E(\nu_t|Y_t)-1]\}.
\end{eqnarray}	

The conditional expectation on the right side of (\ref{cond_exp_gar}) can be computed explicitly for some special cases. Let us consider the Poisson and gamma GLMs. For the Poisson scenario, we can show that the conditional distribution of $\nu_t$ given $Y_t$ follows a gamma distribution with shape and rate parameters $y_t+1/\sigma^2$ and $\mu_t+1/\sigma^2$, respectively, which gives us that $E(\nu_t|Y_t)=\dfrac{y_t+1/\sigma^2}{\mu_t+1/\sigma^2}$. Under the gamma setup, it can be shown that the conditional distribution of $\nu_t|Y_t$ follows a generalised inverse Gaussian distribution with parameters $a=2/\sigma^2$, $b_t=\dfrac{2y_t}{\mu_t\phi}$, and $p=1/\sigma^2-1/\phi$, with density function assuming the form $f_{\nu_t|Y_t}(\nu_t|y_t)=\dfrac{(a/b_t)^{p/2}}{2\mathcal K_p(\sqrt{ab_t})}\nu_t^{p-1}\exp\{-(a\nu_t+b_t/\nu_t)/2\}$, for $\nu_t>0$, where $\mathcal K_p(u)=\dfrac{1}{2}\displaystyle\int_0^\infty \omega^{p-1}\exp\left\{-\dfrac{u}{2}\left(\omega+\dfrac{1}{\omega}\right)\right\}d\omega$ is the modified Bessel function of third kind for $u>0$ and $p\in\mathbb R$. Using this result, it follows that $E(\nu_t|Y_t)=\sqrt{\dfrac{\sigma^2y_t}{\phi\mu_t}}\dfrac{\mathcal K_{1/\sigma_2-1/\phi+1}\left(2\sqrt{\dfrac{y_t}{\sigma^2\phi\mu_t}}\right)}{\mathcal K_{1/\sigma_2-1/\phi}\left(2\sqrt{\dfrac{y_t}{\sigma^2\phi\mu_t}}\right)}$.

Now, consider the ARCH(1) latent process. Then, $E(\nu_{t+l}|\nu_t)=1+\rho^l(\nu_t-1)$, and therefore the conditional expectation $E(Y_{t+l}|Y_t)$ assumes the form (\ref{cond_exp_gar}). In this case, closed-forms expressions for the conditional expectation on the right side of (\ref{cond_exp_gar}) cannot be obtained, but they can be computed aproximately  via Monte Carlo simulation. We have that
\begin{eqnarray*}
	E(\nu_t|Y_t)=\dfrac{\displaystyle\int_0^\infty\nu_tf_{y_t|\nu_t}(y_t|\nu_t)f_{\nu_t}(\nu_t)d\nu_t}{\displaystyle\int_0^\infty f_{y_t|\nu_t}(y_t|\nu_t)f_{\nu_t}(\nu_t)d\nu_t}=\dfrac{E_{\nu_t}\left(\nu_tf_{y_t|\nu_t}(y_t|\nu_t)\right)}{E_{\nu_t}\left(f_{y_t|\nu_t}(y_t|\nu_t)\right)},
\end{eqnarray*}	
where the underscript in $E_{\nu_t}(\cdot)$ means that the expected value is taken with respect to $\nu_t$. Hence, we generate $m$ draws from ARCH trajectories $\{\nu_t^{(i)}\}_{t=1}^n$ ($i=1,\ldots,m$) and approximate the conditional expectations by their Monte Carlo estimates
\begin{eqnarray*}
	\widehat{E}(\nu_t|Y_t)=\dfrac{m^{-1}\displaystyle\sum_{k=1}^m\nu_t^{(k)}f_{y_t|\nu_t}(y_t|\nu_t^{(k)})}{m^{-1}\displaystyle\sum_{k=1}^mf_{y_t|\nu_t}(y_t|\nu_t^{(k)})},
\end{eqnarray*}	
for $t=1,\ldots,n$. The methodology proposed here for prediction will be illustrated in the real data applications in Section \ref{sec:applications}.


\section{Real data analyses}\label{sec:applications}

We here present two real data applications to show the performance of our GLM time series models in practice to handle counts and positive continuous data.

\subsection{Measles infections in North Rhine-Westphalia (Germany)}
Let us consider the weekly number of measles infections reported in North Rhine-Westphalia (Germany) from January 2001 to May 2013, so totaling 646 observations. The data is publicly available in the \texttt{R} package \texttt{tscount}. The plot of the measles time series and its ACF are displayed in Figure \ref{fig:measles_plot}. The following vector of covariates, including trend and seasonal components, are considered:
\begin{eqnarray*}
	{\bf x}_t=\left(1,t/646,\cos(2\pi t/52),\sin(2\pi t/52),\cos(4\pi t/52),\sin(4\pi t/52),\cos(8\pi t/52),\sin(8\pi t/52)\right)^\top,	
\end{eqnarray*}
for $t=1,\ldots,646$. 

\begin{figure}[h!]
	\centering
	\includegraphics[width = .5\linewidth]{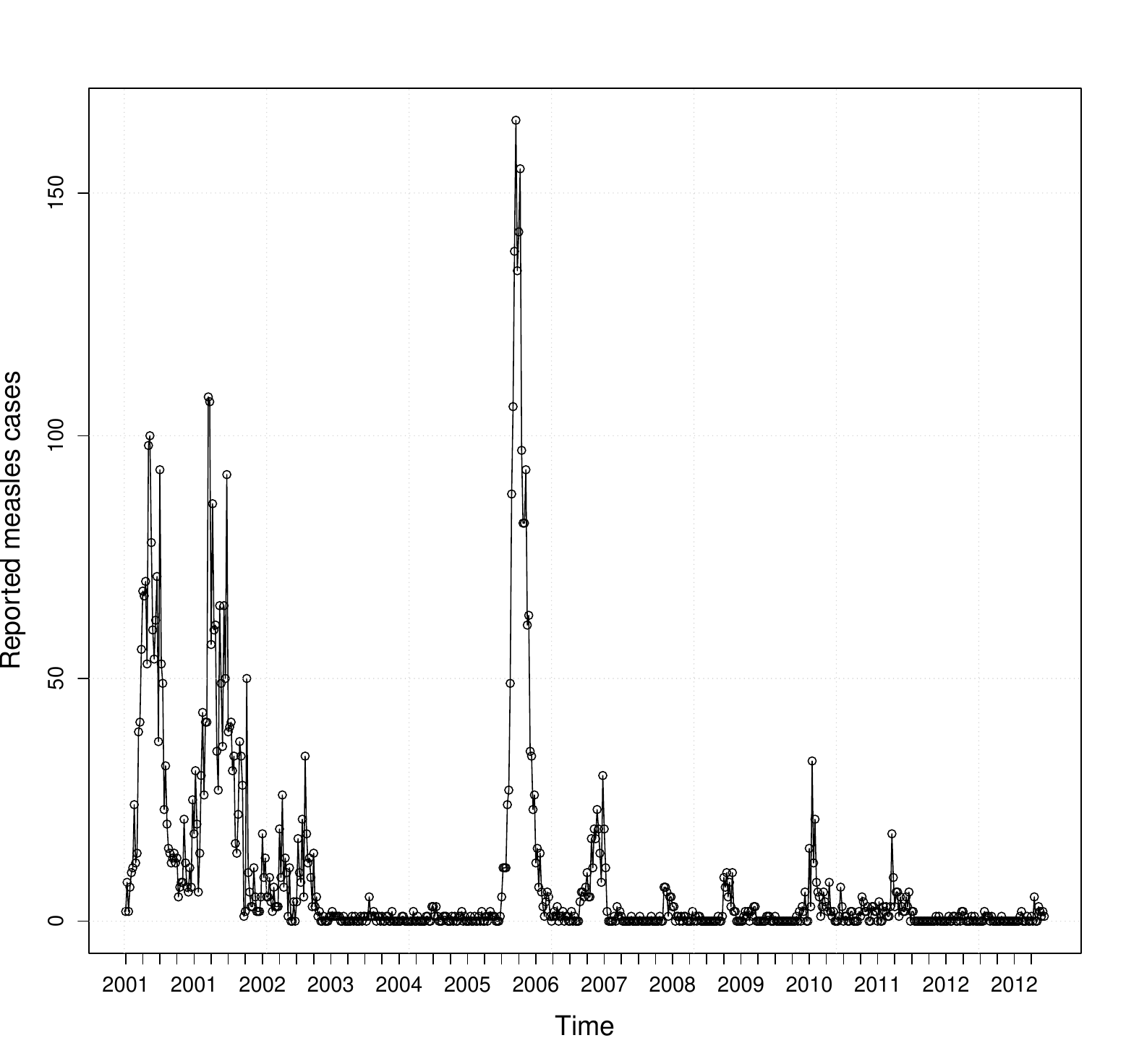}\includegraphics[width = .5\linewidth]{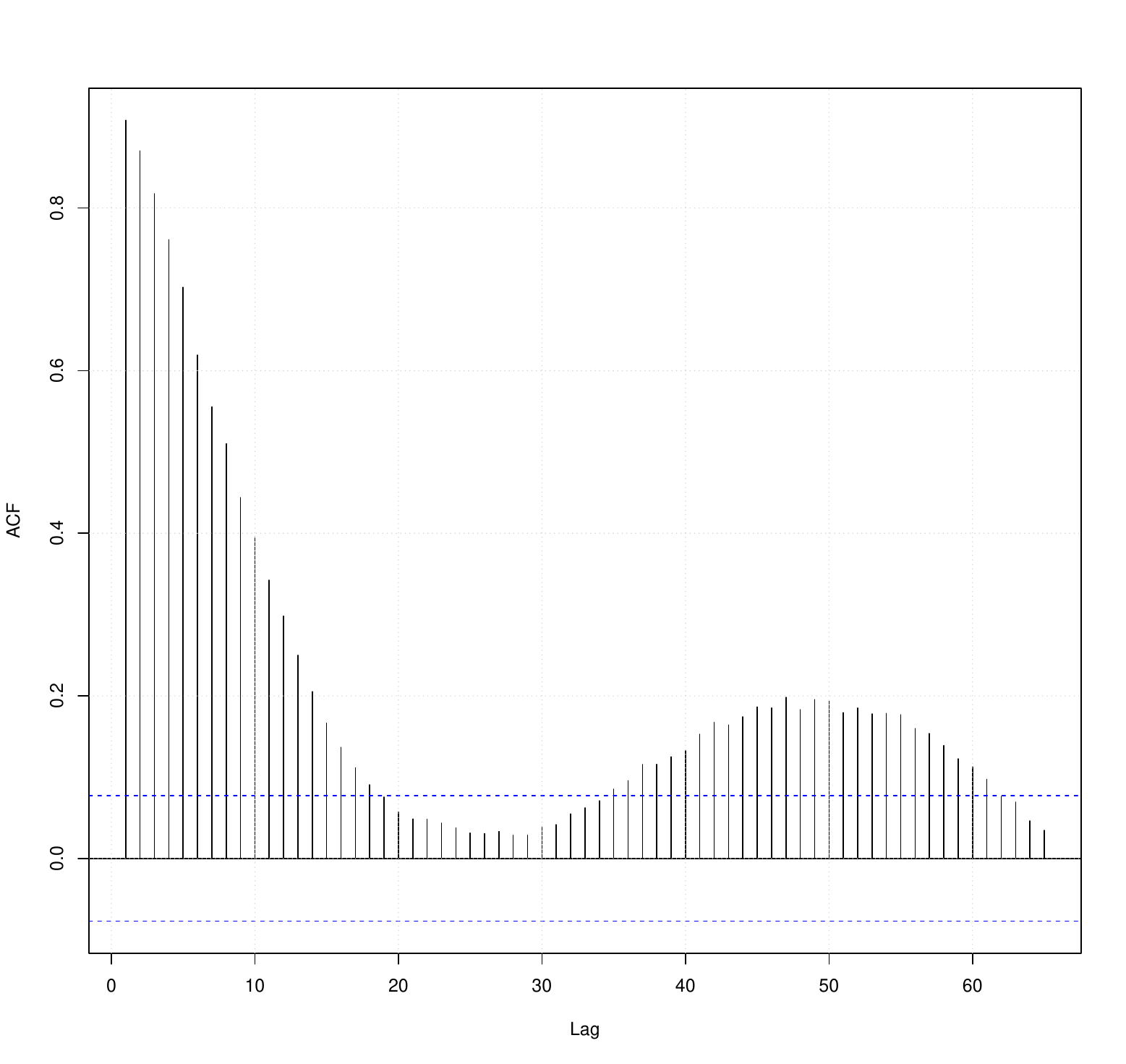}
	\caption{Plots of the weekly reported cases of measles in North Rhine-Westphalia, Germany, between January 2001 and May 2013 and its respective autocorrelation function.}\label{fig:measles_plot}
\end{figure}

Table \ref{tab:inf_measles} provides a summary of the GLM fitting with the respective standard errors. The estimation of the parameters $\sigma^2$ and $\rho$ of the Poisson LNAR(1) model are obtained through the method of moment procedure:
\begin{eqnarray*}
	\widehat\sigma^2&=&\log\left(\dfrac{\sum_{t=1}^n[(Y_t-\widehat\mu_t)^2-\widehat\mu_t]}{\sum_{t=1}^n\widehat\mu_t^2}+1\right),\quad
	\widehat\rho=\log\left(\dfrac{\sum_{t=1}^n(Y_t-\widehat\mu_t)(Y_{t-1}-\widehat\mu_{t-1})}{\sum_{t=2}^n\widehat\mu_t\widehat\mu_{t-1}}+1\right)\big/\widehat\sigma^2,
\end{eqnarray*}	
$\widehat\mu_t=h({\bf x}_{nt}^\top\widehat{\boldsymbol\beta})$, with  $\widehat{\boldsymbol\beta}$ being the GLM estimator of $\boldsymbol\beta$. In this case, the estimates are $\widehat\sigma^2=0.751$ and $\widehat\rho=0.924$. As for the Poisson GAR(1) model, the method of moments estimators of $\sigma^2$ and $\rho$ assume the form
\begin{eqnarray*}
	\widehat\sigma^2&=&\dfrac{\sum_{t=1}^n[(Y_t-\widehat\mu_t)^2-\widehat\mu_t]}{\sum_{t=1}^n\widehat\mu_t^2},\quad
	\widehat\rho=\dfrac{\sum_{t=1}^n(Y_t-\widehat\mu_t)(Y_{t-1}-\widehat\mu_{t-1})}{\widehat\sigma^2\sum_{t=2}^n\widehat\mu_t\widehat\mu_{t-1}},
\end{eqnarray*}	
which yield the estimates $\widehat\sigma^2=1.118$ and $\widehat\rho=0.895$.

Under the Poisson ARCH(1) model, to estimate the parameter $\rho$, we use the first-order autocorrelation of $Y_t$ to obtain the method of moments estimator $\widehat\rho$ as the solution of the nonlinear equation
\begin{eqnarray*}
	\widehat\rho\left(\dfrac{3(1-\widehat\rho^2)}{1-3\widehat\rho^2}-1\right)=\displaystyle\sum_{t=1}^{n-1}(Y_t-\widehat\mu_t)(Y_{t+1}-\widehat\mu_{t+1})/\displaystyle\sum_{t=1}^{n-1}\widehat\mu_t\widehat\mu_{t+1},
\end{eqnarray*}	
The method of moments estimate of $\rho$ is $\widehat\rho=0.333$, which belongs to the interval $(0,1/\sqrt{3})$. Therefore, the theoretical results obtained in this paper can be applied to analyze the measles time series.

To obtain the standard errors under the Poisson case ($V(\mu)=\mu$), we use the fact that $E(\nu_t)=1$ and hence both $\boldsymbol\Omega_{I}$ and $\boldsymbol\Omega^\dag_{I}$ can be consistently estimated by $\widehat{\boldsymbol\Omega}_{I}=\widehat{\boldsymbol\Omega}^\dag_{I}=n^{-1}\displaystyle\sum_{t=1}^n{\bf x}_{nt}{\bf x}_{nt}^\top\widehat\mu_t$, while $\widehat{\boldsymbol\Omega}_{II}=n^{-1}\displaystyle\sum_{t=1}^n\sum_{k=1}^n{\bf x}_{nt}{\bf x}_{nk}^\top\widehat\mu_t\widehat\mu_k\widehat\gamma_\nu(|t-k|)$ is a consistent estimator for ${\boldsymbol\Omega}_{II}$, where $\widehat\gamma_\nu(\cdot)$ is a consistent estimator of $\gamma_\nu(\cdot)$.

In addition, a parametric bootstrap with 1000 replications is also employed to obtain the standard errors and the related results are also presented in Table  \ref{tab:inf_measles}. As expected, the GLM approach does not provide correct standard errors as it ignores the temporal dependence. As a consequence, we notice the difference between such a quantity provided under GLM and bootstrap approach, which is line with the numerical experiments by \cite{davetal2000} and \cite{davwu2009}. Only the covariate $\sin(8\pi t/52)$ is not significant according to GLM-based inference at a significance level at 5\%. On the other hand, the covariates $\cos(4\pi t/52)$, $\sin(4\pi t/52)$, $\cos(8\pi t/52)$, and $\sin(8\pi t/52)$ are not significant based on the parametric bootstrap method, where a normal approximation was considered, which is in line with the inference based on correct information matrix, where the temporal dependence is taken into account.

\begin{table}[h!]
	\centering
	\begin{tabular}{|c|cccccccc|}
		\hline
		estimates & $\beta_0$ & $\beta_1$ & $\beta_2$ & $\beta_3$ & $\beta_4$ & $\beta_5$ & $\beta_6$ & $\beta_7$ \\
		\hline
		GLM   & 3.043     & $-$3.370  & $-$0.683  & 1.108     & $-$0.054  & $-$0.083  & $-$0.040  & $-$0.012  \\
		Boot. LNAR& 2.947& $-$3.296& $-$0.675& 1.108& $-$0.054& $-$0.079& $-$0.037& $-$0.011\\
		Boot. GAR&2.861& $-$3.145& $-$0.728& 1.070& $-$0.078& $-$0.095& $-$0.037& $-$0.022\\
		Boot. ARCH& 2.962& $-$3.328& $-$0.684& 1.108& $-$0.066& $-$0.078& $-$0.035& $-$0.007\\
		\hline
		stand. error & $\beta_0$ & $\beta_1$ & $\beta_2$ & $\beta_3$ & $\beta_4$ & $\beta_5$ & $\beta_6$ & $\beta_7$ \\
		\hline
		GLM     & 0.025     & 0.057     & 0.027     & 0.029     & 0.023     & 0.023     & 0.019     & 0.019     \\
		LNAR& 0.441& 0.981& 0.216& 0.221& 0.148& 0.150& 0.097& 0.097\\
		Boot. LNAR& 0.391 &0.879& 0.194& 0.201& 0.129& 0.130& 0.084& 0.082\\			 
		GAR& 0.418& 0.946& 0.225& 0.229& 0.153& 0.155& 0.098& 0.098\\	
		Boot. GAR&0.378& 0.821& 0.220& 0.224& 0.147& 0.157& 0.102& 0.095\\	 			 
		ARCH&0.248& 0.604& 0.185& 0.182& 0.207& 0.205& 0.214& 0.214\\ 
		Boot. ARCH    & 0.212& 0.534& 0.171& 0.164& 0.188& 0.185& 0.194& 0.185   \\
		\hline
	\end{tabular}
	\caption{Parameter estimates based on GLM, bootstrap, and their respective standard errors for the measles time series.}
	\label{tab:inf_measles}
\end{table}

To evaluate the predictive performance of the models, we compute the root mean square error (RMSE) based in-sample forecasting. Table \ref{tab:pred_measles} provides the RMSE and the correlation between predictions and observations for the models considered here. In-sample forecasting for the weekly reported cases of measles based on GLM, GLM-LNAR, GLM-GAR, and GLM-ARCH along with the observed time series are displayed in Figure \ref{fig:measles_pred}. These results show that GLM-LNAR and GLM-GAR yield good results in terms of prediction with a better performance for the GLM-based on the GAR latent process.

\begin{table}[h!]
	\centering
	\begin{tabular}{|c|cccc|}
		\hline
		Model$\rightarrow$ & GLM    & GLM-LNAR & GLM-GAR & GLM-ARCH \\
		\hline
		RMSE               & 17.761 &  8.837        &   8.724      & 12.893   \\
		Correlation        & 0.582  &  0.914        &  0.917       & 0.820 \\
		\hline
	\end{tabular}
	\caption{RMSE and correlation between observations and their respective in-sample predictions under GLM, GLM-LNAR, GLM-GAR, and GLM-ARCH for the measles time series data analysis.}
	\label{tab:pred_measles}
\end{table}

\begin{figure}[h!]
	\centering
	\includegraphics[width = .5\linewidth]{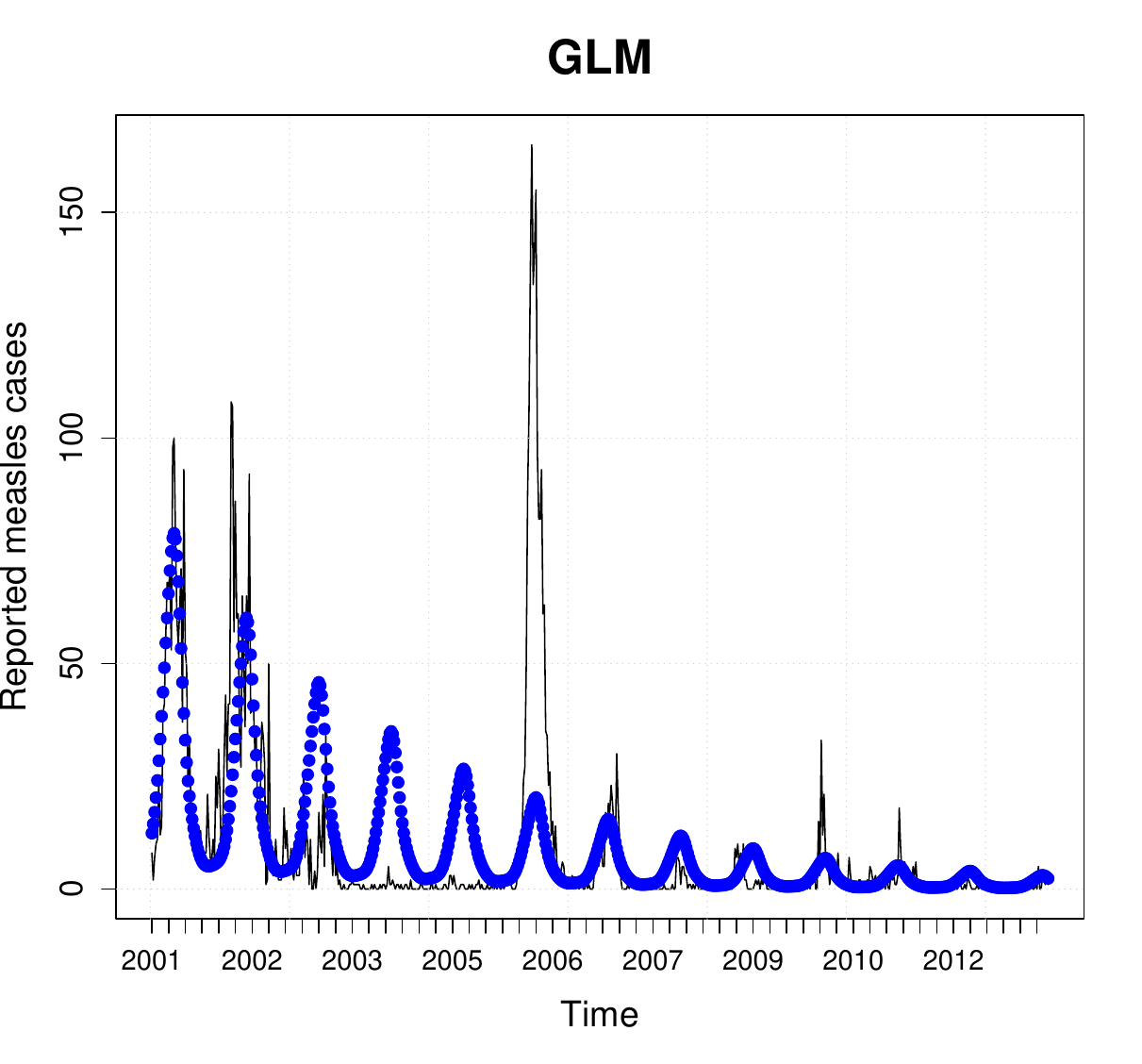}\includegraphics[width = .5\linewidth]{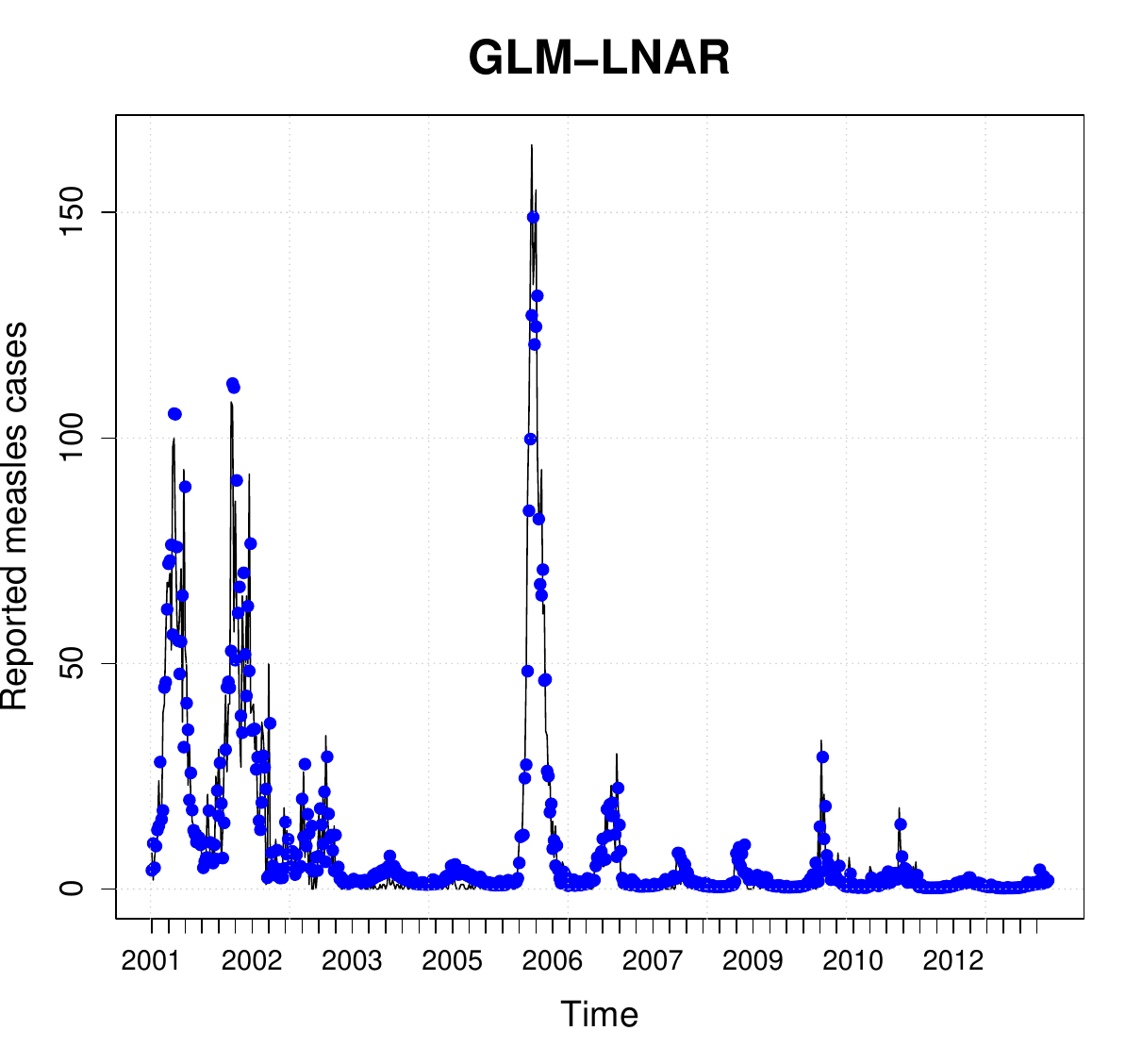}
	\includegraphics[width = .5\linewidth]{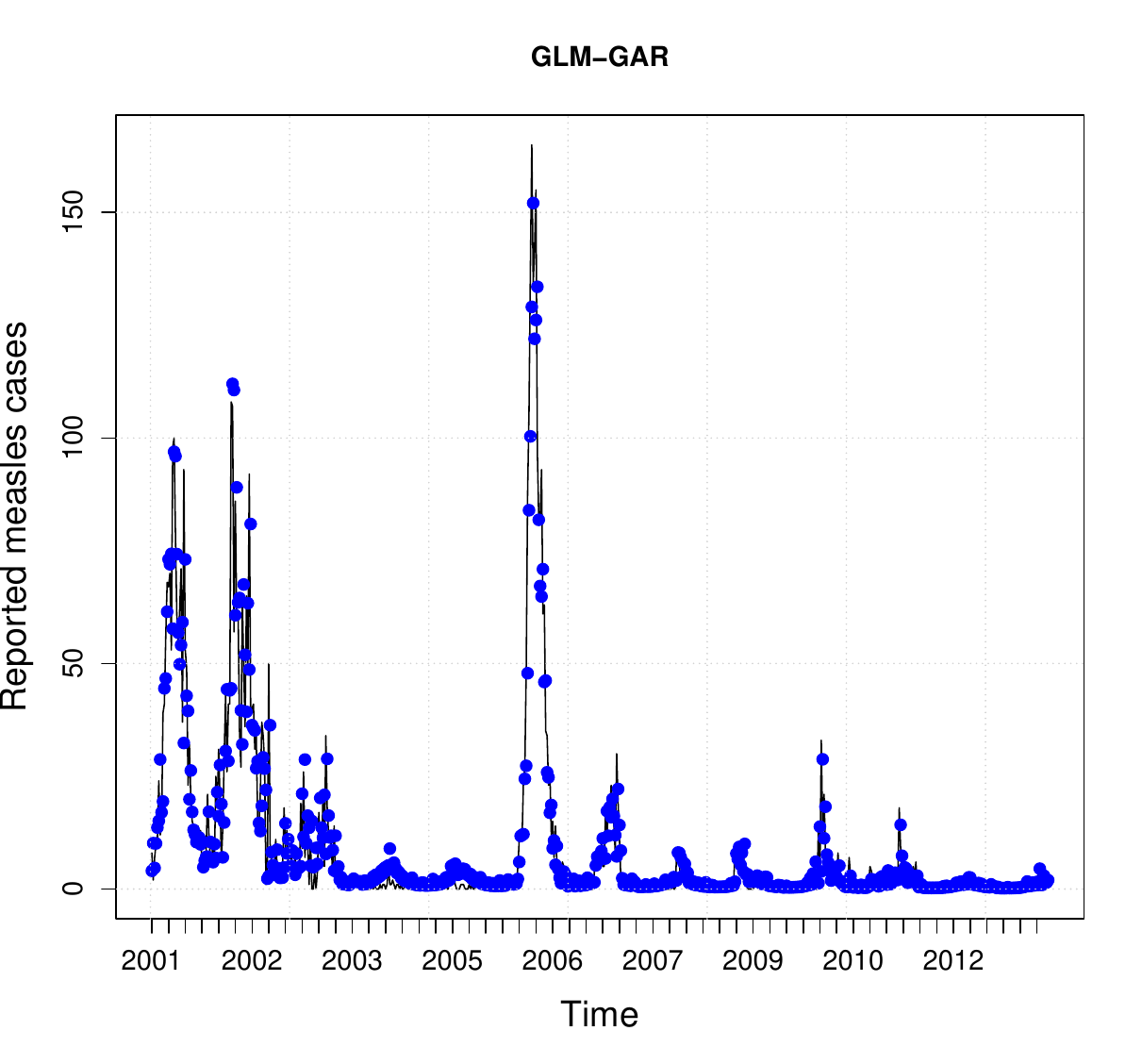}\includegraphics[width = .5\linewidth]{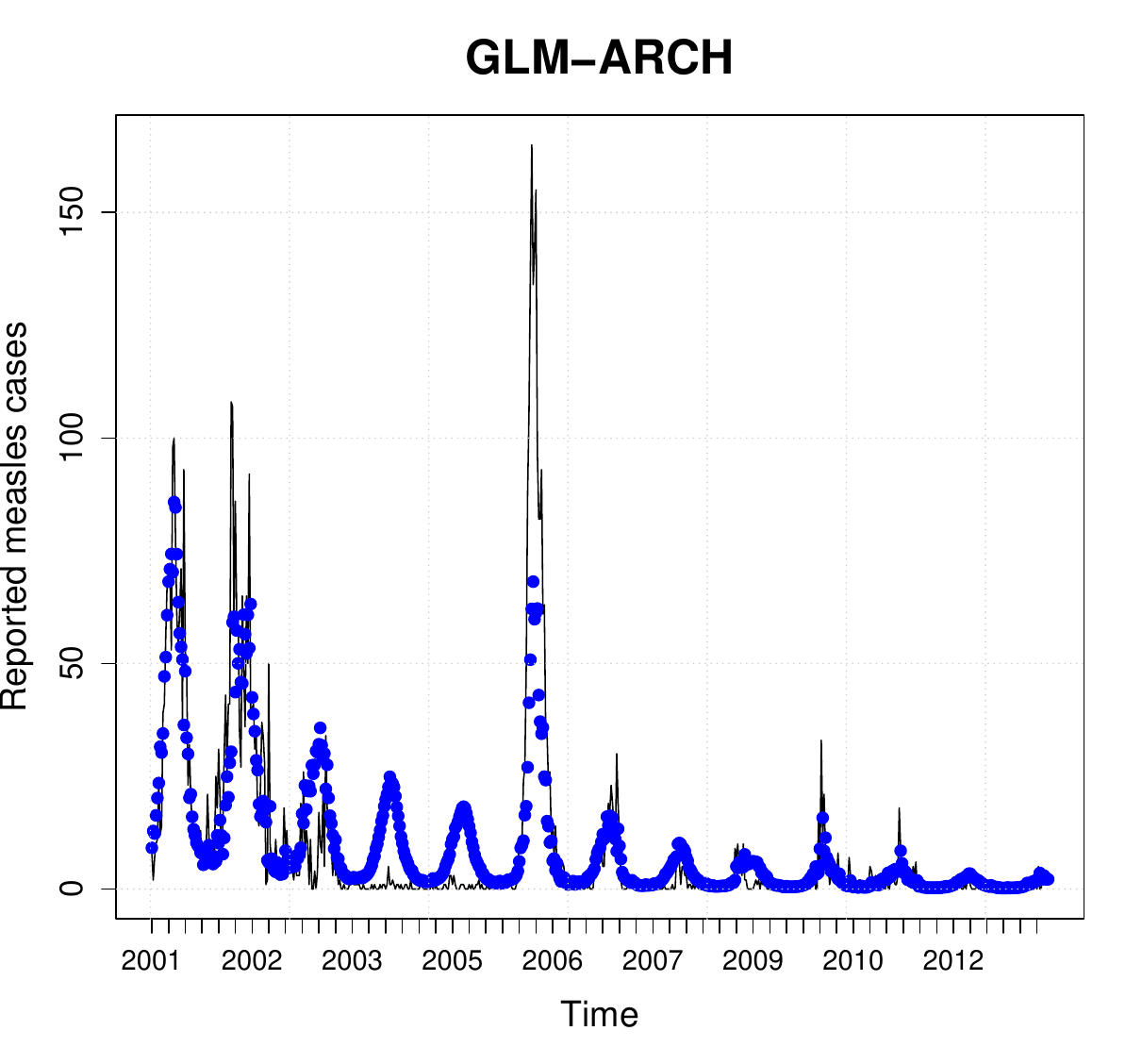}
	\caption{In-sample forecasting (dots) for the weekly reported cases of measles based on GLM, GLM-LNAR, GLM-GAR, and GLM-ARCH along with the observed time series (solid line).}\label{fig:measles_pred}
\end{figure}

\subsection{Paleoclimatic glacial varves}

The second real data application is about the ticknesses of the yearly sedimentary deposits (varves) from one location in Massachusetts (about 12,600 years ago) for 634 years. Varves can be used as proxies for paleoclimatic parameters, for instance, to reconstruct past climates and other environmental changes. For more details, see Example 2.6 from the book  by \cite{shusto2011}. A plot of the varve time series and its ACF are provided in Figure \ref{fig:varve_plot}.  We consider a gamma GLM (with an inverse canonical link function) driven by latent processes and vector of covariates with a trend: ${\bf x}_t=\left(1,t/634\right)^\top$, for $t=1,\ldots,634$.

\begin{figure}[h!]
	\centering
	\includegraphics[width = .5\linewidth]{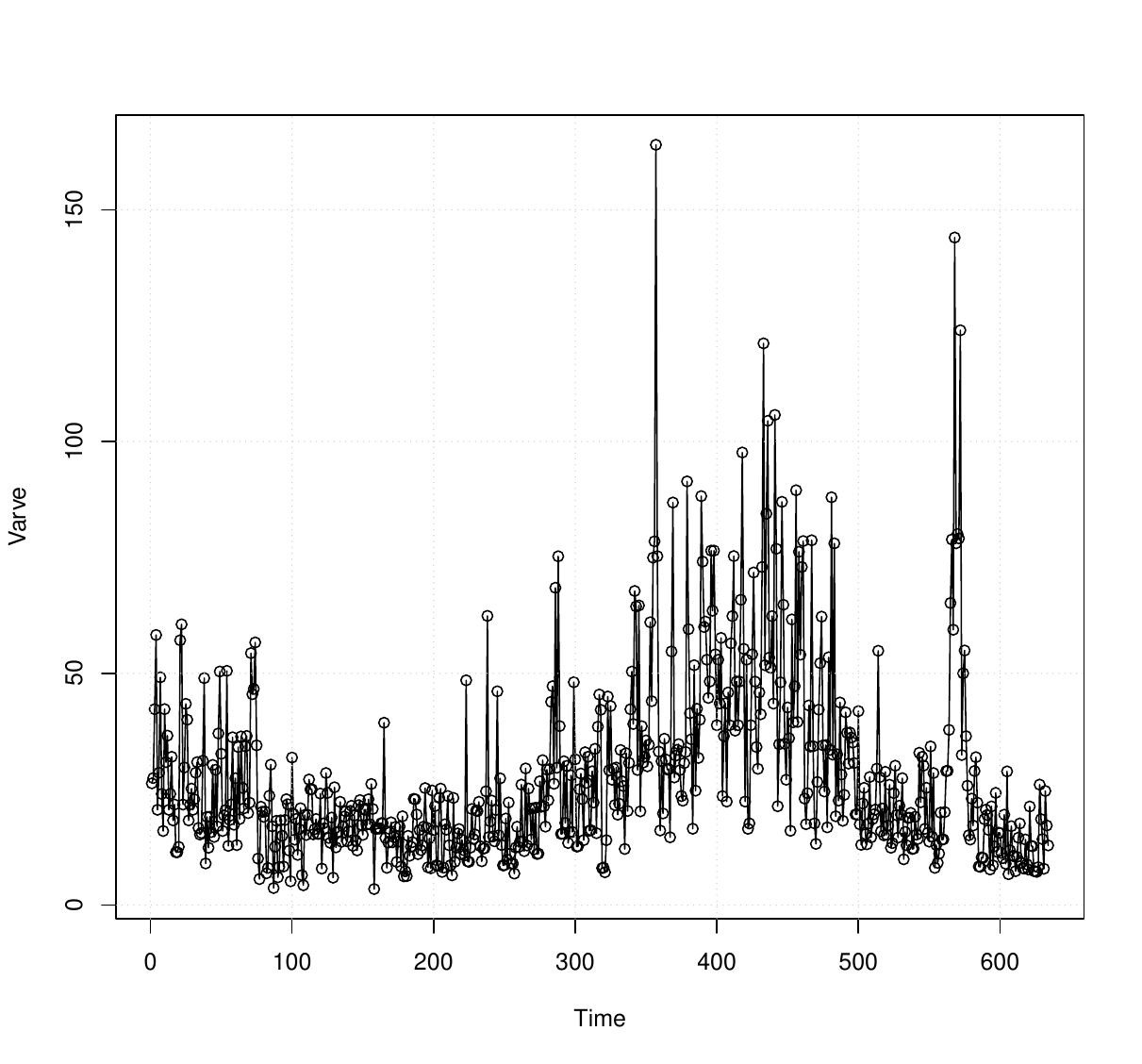}\includegraphics[width = .5\linewidth]{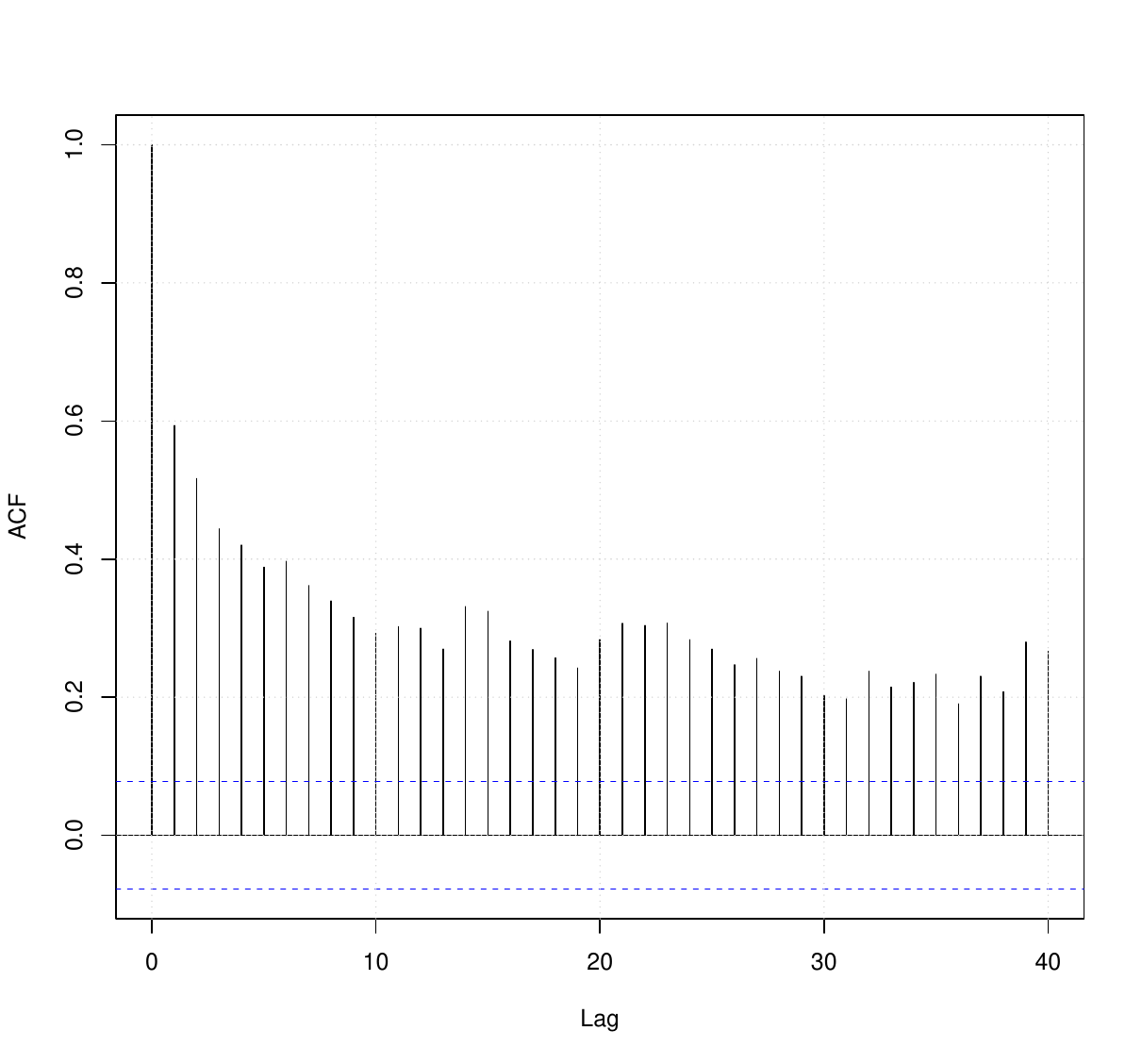}
	\caption{Plots of the varve time series and its respective autocorrelation function.}\label{fig:varve_plot}
\end{figure}

The summary fit of the gamma GLM under LNAR(1), GAR(1), and ARCH(1) latent processes are provided in Table \ref{tab:inf_varve}. To estimate the dispersion $\phi$ parameters of the latent processes, we use method of moments. Under LNAR(1), we use the estimators
\begin{eqnarray*}
	\widehat\rho&=&\log\left(\dfrac{\sum_{t=3}^n(Y_t-\widehat\mu_t)(Y_{t-2}-\widehat\mu_{t-2})}{\sum_{t=3}^n\widehat\mu_{t}\widehat\mu_{t-2}}+1\right)\bigg/\log\left(\dfrac{\sum_{t=2}^n(Y_t-\widehat\mu_t)(Y_{t-1}-\widehat\mu_{t-1})}{\sum_{t=2}^n\widehat\mu_{t}\widehat\mu_{t-1}}+1\right),\\
	\widehat\sigma^2&=&\log^2\left(\dfrac{\sum_{t=2}^n(Y_t-\widehat\mu_t)(Y_{t-1}-\widehat\mu_{t-1})}{\sum_{t=2}^n\widehat\mu_{t}\widehat\mu_{t-1}}+1\right)\bigg/\log\left(\dfrac{\sum_{t=3}^n(Y_t-\widehat\mu_t)(Y_{t-2}-\widehat\mu_{t-2})}{\sum_{t=3}^n\widehat\mu_{t}\widehat\mu_{t-2}}+1\right),\\
	\widehat\phi&=&\exp(-\widehat\sigma^2)\left(\dfrac{\sum_{t=1}^n(Y_t-\widehat\mu_t)^2}{\sum_{t=1}^n\widehat\mu_t^2}+1\right)-1,
\end{eqnarray*}	
which yield the parameter estimates $\widehat\rho=0.881$, $\widehat\sigma^2=0.297$, and $\widehat\phi=0.123$. As for the GAR(1) latent process, we have the estimators
\begin{eqnarray*}
	\widehat\rho&=&\dfrac{\sum_{t=3}^n(Y_t-\widehat\mu_t)(Y_{t-2}-\widehat\mu_{t-2})}{\sum_{t=2}^n(Y_t-\widehat\mu_t)(Y_{t-1}-\widehat\mu_{t-1})}
	\dfrac{\sum_{t=2}^n\widehat\mu_{t}\widehat\mu_{t-1}}{\sum_{t=3}^n\widehat\mu_{t}\widehat\mu_{t-2}},\\	
	\widehat\sigma^2&=&\left(\dfrac{\sum_{t=2}^n(Y_t-\widehat\mu_t)(Y_{t-1}-\widehat\mu_{t-1})}{\sum_{t=2}^n\widehat\mu_{t}\widehat\mu_{t-1}}\right)^2\dfrac{\sum_{t=3}^n\widehat\mu_{t}\widehat\mu_{t-2}}{\sum_{t=3}^n(Y_t-\widehat\mu_t)(Y_{t-2}-\widehat\mu_{t-2})},\\
	\widehat\phi&=&(\widehat\sigma^2+1)^{-1}\left(\dfrac{\sum_{t=1}^n(Y_t-\widehat\mu_t)^2}{\sum_{t=1}^n\widehat\mu_t^2}+1\right)-1,
\end{eqnarray*}	
with the parameter estimates $\widehat\rho=0.867$, $\widehat\sigma^2=0.345$, and $\widehat\phi=0.123$. In this application, we do not consider the ARCH latent process since that the MM estimator of $\phi$ yielded a negative value.

We can notice a good agreement between the GLM and bootstrap estimates and standard errors (under our formulation) as well. By ignoring the dependence on the data, the trend is statistically significant, which is conflitant with the conclusion based on the modelling taking into account dependence, where this covariate is not significant.


\begin{table}[h!]
	\centering
	\begin{tabular}{|c|cc|}
		\hline
		estimates & $\beta_0$ & $\beta_1$   \\
		\hline
		GLM   &   0.044  &    $-$0.016       \\
		Boot. LNAR&  0.045 & $-$0.016  \\
		Boot. GAR&  0.046 & $-$0.018 \\
		\hline
		stand. error & $\beta_0$ & $\beta_1$   \\
		\hline
		GLM   & 0.002  & 0.003\\
		LNAR&  0.008 & 0.012 \\
		Boot. LNAR&  0.007 & 0.011 \\
		GAR&  0.008 & 0.012\\
		Boot. GAR&  0.007& 0.011\\	 
		\hline
	\end{tabular}
	\caption{Parameter estimates based on GLM, bootstrap, and their respective standard errors for the varve time series.}
	\label{tab:inf_varve}
\end{table}

Table \ref{tab:pred_varve} reports the RMSE and correlation between observations and their respective in-sample predictions under GLM, GLM-LNAR, and GLM-GAR for the varve time series. GLM-LNAR and GLM-GAR produce the best and similar results, with slight advantage for the latter.  Plots of the in-sample predictions based on GLM, GLM-LNAR, and GLM-GAR are presented in Figure \ref{fig:varve_pred}, which shows a good performance of the prediction methodology of our models proposed in Section \ref{sec:prediction}.

\begin{table}[h!]
	\centering
	\begin{tabular}{|c|ccc|}
		\hline
		Model$\rightarrow$ & GLM    & GLM-LNAR & GLM-GAR\\ 
		\hline
		RMSE               & 20.099   &   16.098  & 16.065  \\
		Correlation           & 0.148   &  0.610     &  0.612  \\ 
		\hline
	\end{tabular}
	\caption{RMSE and correlation between observations and their respective in-sample predictions under GLM, GLM-LNAR, and GLM-GAR for the varve time series data analysis.}
	\label{tab:pred_varve}
\end{table}

\begin{figure}[h!]
	\centering
	\includegraphics[width = .5\linewidth]{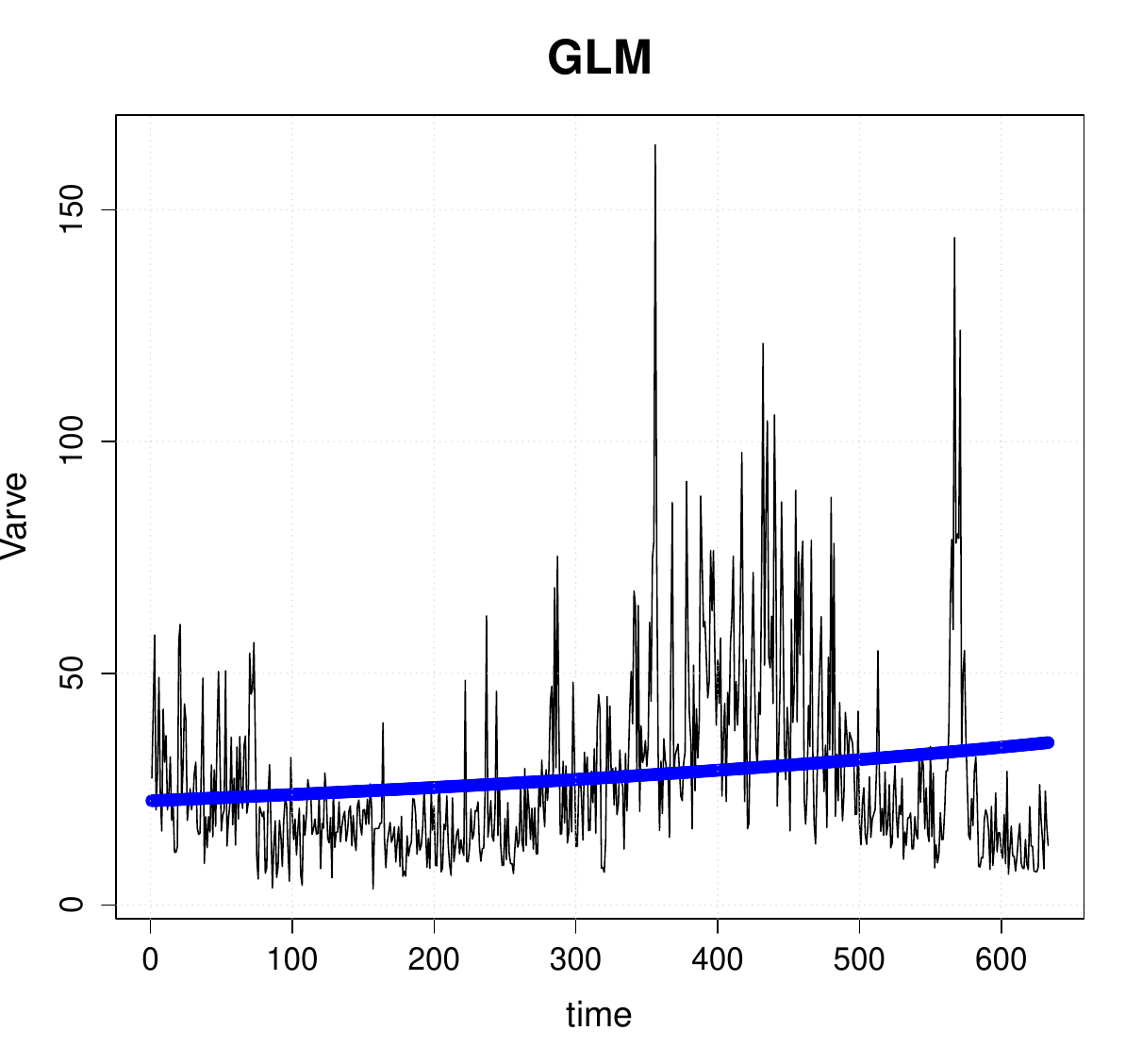}\includegraphics[width = .5\linewidth]{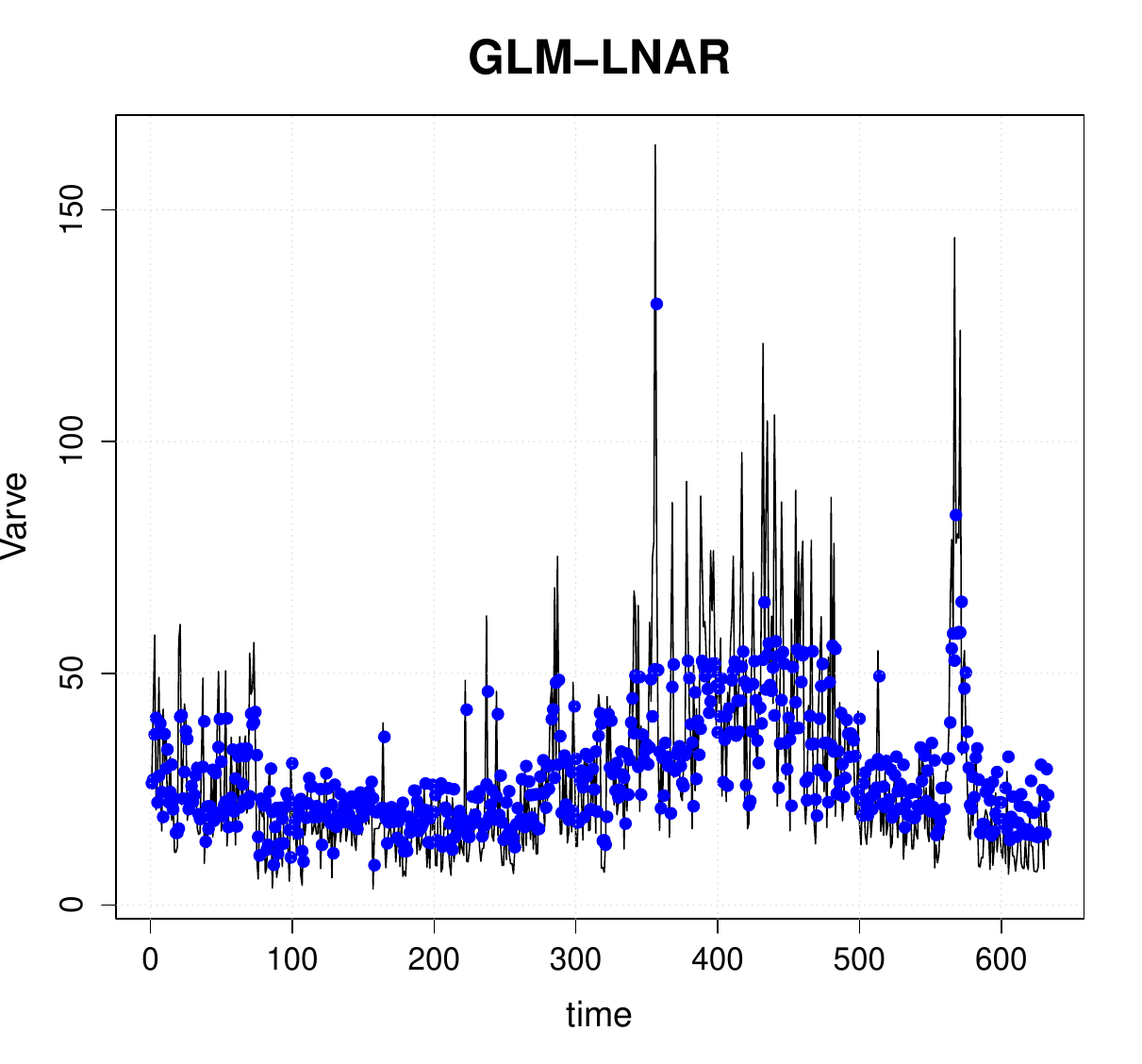}
	\includegraphics[width = .5\linewidth]{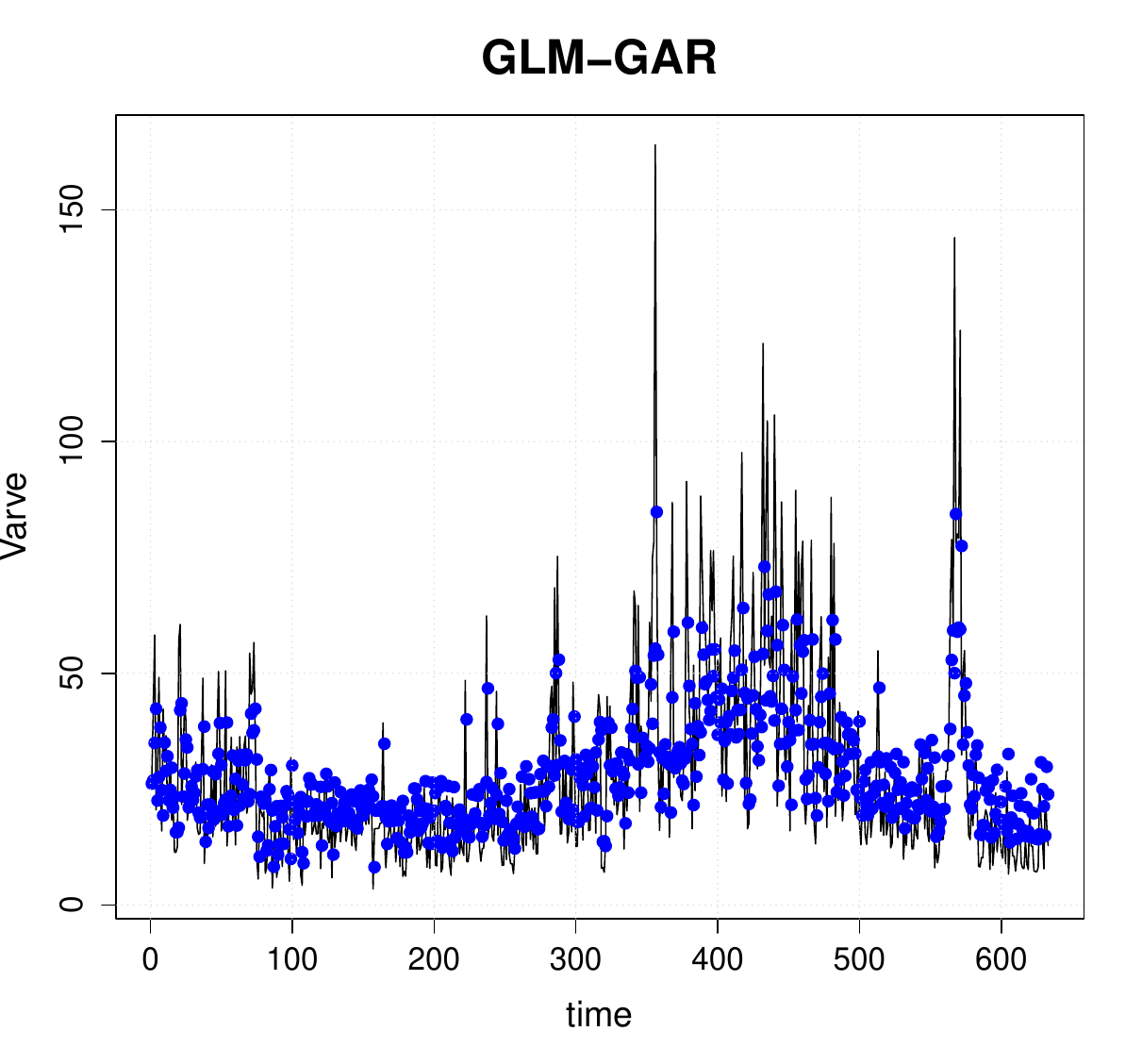}
	\caption{In-sample forecasting (dots) for the varve time series based on GLM, GLM-LNAR, GLM-GAR, and GLM-ARCH along with the observed time series (solid line).}\label{fig:varve_pred}
\end{figure}

\section{Concluding remarks}\label{sec:conclusion}

We proposed a flexible class of generalised linear models driven by latent processes for modeling count, real-valued, continuous, and positive continuous time series. By assuming that ${Y_t}$, conditional on ${\nu_t}$, follows a bi-parameter exponential family with a multiplicative latent effect in the conditional mean, our formulation overcomes some limitations of existing approaches and accommodates important cases, such as gamma-distributed positive time series, that are not feasible under the additive latent structure of \cite{davwu2009}. The inclusion of an estimable dispersion parameter and the absence of a strong restriction condition on the link function $h(\cdot)$ further enhance model flexibility.

We established the asymptotic normality of the GLM estimators using the Central Limit Theorem for strongly mixing processes by \cite{pelute1997} and derived the corresponding information matrix for inference. We also developed prediction methods for GLMs with latent processes, addressing a gap in the literature. Empirical illustrations showed that alternative latent dynamics, such as a gamma AR(1) process, may provide improved performance over the commonly used log-normal AR(1). Overall, the proposed framework offers a versatile and theoretically sound basis for modelling time series with latent dynamic structure. 

\section*{Acknowledgements}

\noindent WBS thanks the Department of Biostatistics from City University of Hong Kong for the warm hospitality during his research visit, where part of this work was developed.

\end{document}